\documentclass{IEEEtran}

\usepackage{graphics}
\usepackage{amsmath}
\usepackage{amssymb}
\usepackage{epsfig,wrapfig}
\usepackage{amsmath}
\usepackage{color}
\usepackage{float}
\usepackage{graphicx}
\usepackage{relsize}
\usepackage{cite}

\newtheorem{definition}{Definition}
\newtheorem{theorem}{Theorem}

\title{Averager-copier-voter models for hybrid opinion dynamics in complex networks\thanks{The authors are with the School of Electrical Engineering and Computer Science at Washington State University. This work was partially supported by United States National Science Foundation Grants 1545104 and 1635184, and the Homeland Security Advanced Research Projects Agency (HSARPA).}}

\author{Mengran Xue and Sandip Roy}

\begin{document}

\maketitle

\begin{abstract}
A hybrid model for opinion dynamics in complex multi-agent networks is introduced, wherein some continuous-valued agents average neighbors' opinions to update their own, while other discrete-valued agents use stochastic copying and voting protocols.  A statistical and graph-theoretic analysis of the model is undertaken, and consensus is shown to be achieved whenever the network matrix  is ergodic.  Also,  the time required for consensus is characterized, in terms of the network's graph and the distribution of agents of different types.
\end{abstract}

\section{Introduction}

Synchronization or consensus processes occur in myriad natural-world and engineered networks.  A number of models have been proposed for synchronization phenomena in the physical world, such as formation flight of birds and phase alignment of oscillator circuits \cite{sync1,sync2,sync3,sync4,sync5}.  These models describe local interaction mechanisms that induce network components' continuous-valued states to come to a common value over time. In a parallel track, voter and influence models have been introduced to describe consensus among networked cyber or human agents with discrete-valued opinions (e.g., distributed decision-making algorithms in computer networks, elections or gossip in a community)\cite{voter1,voter2,voter3,voter4,voter5,voter6,voter7,voter8}.    These models, and their attendant graph-theoretic analyses, have proved to be keen tools for predicting and designing the dynamics of complex networks (e.g. \cite{dhal}).

Increasingly, complex systems comprise tightly interconnected cyber, physical, and human elements -- consider for instance microrobot swarms that move through the bloodstream, or computer-assisted decision-making in infrastructure control centers \cite{gill,cpsbook1,cpsbook2,cps3}.  These cyber-physical networks inevitably involve hybridized decision-making dynamics, with some components or agents governed by continuous-valued laws (e.g., diffusion equations) and others engaged in discrete behaviors (e.g., voting).  Even for complex systems without these heterogeneities, hybridized decision-making dynamics may be of interest -- for instance, voting processes in social networks may involve a mixture of players with graded and binary opinions, and building thermal processes are continuous-valued but are subject to discretized sensing and actuation.   These contexts motivate the study of synchronization or consensus processes among agents with hybrid states  and interactions.

In this brief paper, a stochastic network model for consensus among mixed continuous-valued and discrete-valued agents is introduced. The model is not drawn from a particular application domain, but rather is an effort to abstractly represent essential hybrid interactions that are involved in decision-making or consensus. Specifically, the model encapsulates agents that update their states or opinions in three different ways: linear averaging \cite{sync3,lincon1,lincon2,lincon3,lincon4}, stochastic copying, and stochastic voting \cite{voter1,voter2,voter3,voter4,voter5,voter6,voter7}.  A statistical (two-moment) analysis of the model is undertaken, and used to verify that consensus is achieved under broad connectivity conditions.  Also, a spectral and graph-theoretic analysis of the consensus time is undertaken, which shows that hybridized decision-making processes reach consensus slowly compared to pure averaging and pure voting processes. 

We note that preliminary results on the convergence of the hybrid model were presented in \cite{confpap}.  Further studies on hybrid consensus models that follow on our initial work, including for mobile and time-varying networks, can be found in \cite{sheng,sheng2}.

\section{Hybrid Consensus Model: Formulation}

A network with $n$ {\em agents}, labeled $i=1,\hdots, n$, is considered. Each agent has a scalar {\em opinion}
$x_i[k]$ that evolves in discrete time ($k=0,1,\hdots$). 
The initial opinion $x_i[0]$ of each agent $i$ is assumed to be a arbitrary real number in the interval $[0,1]$.   At each time step $k>0$, each agent $i$ updates its opinion based on the opinions of the agents
 in its {\em neighbor set} ${\cal N}_i$ (possibly including itself).  The network is assumed to comprise three types of agents, which have different update rules reflecting discrete vs continuous opinions and update mechanisms:

1) Each of the $m_a$ {\em averager agents} (labeled $1, \hdots, m_a$) computes its next opinion as a weighted average of the current opinions of its neighbors.  Specifically, for an averager agent $i$, the time-$(k+1)$ opinion is computed as $x_i[k+1]=\sum_{j=1}^n g_{ij} x_j[k]$, where $g_{ij} \ge 0$, $\sum_{j=1}^n g_{ij}=1$, and $g_{ij} \neq 0$ only if $j \in {\cal N}_i$.  We note that the averagers maintain a continuous-valued opinion, and compute a continuous-valued function to determine their next opinion.

2)  Each of the $m_c$ {\em copier agents} (labeled $m_a+1,\hdots, m_a+m_c$) computes its next opinion by stochastically polling or selecting a neighbor, and assuming the opinion of that neighbor.  Specifically, the time-$(k+1)$ opinion of a copier agent $i$ is determined as follows.  First, the agent chooses a {\em determining agent} $j$ with probability $g_{ij}$, where $g_{ij} \ge 0$, $\sum_{j=1}^n g_{ij}=1$, and $g_{ij}>0$ only if $j \in {\cal N}_i$; this selection of the determining agent is independent of all other past and current selections.  The agent then simply copies the status of the determining agent, i.e. $x_i[k+1]=x_j[k]$.
The copiers thus maintain a continuous-valued opinion, but updates this opinion via 
a discrete selection process. 

3) Each of the $m_v$ {\em voter agents} (labeled $m_a+m_c+1,\hdots, n$)  computes its next opinion as a binary choice (0 or 1), based on a
 weighted average of neighbors' opinions.  Specifically, to compute its time-$(k+1)$ opinion, a voter agent $i$ first computes a weighted average 
 $f_i[k+1]=\sum_{j=1}^n g_{ij} x_j[k]$, where $g_{ij} \ge 0$, $\sum_{j=1}^n g_{ij}=1$, and $g_{ij} \neq 0$ only if $j \in {\cal N}_i$.  The agent's next state is 
then determined as $x_i[k+1]=1$ with probability $f_i[k+1]$ and $x_i[k+1]=0$ otherwise, independently of all other current and past updates.  The voters thus can be viewed as computing a continuous function (a weighted average), but then selecting a binary opinion based on this computation.  Alternately, the voter $i$ can be interpreted as first selecting a neighbor $j$ with probability $g_{ij}$, and then setting its next status to $1$ with probability $x_j[k]$ (and to $0$ otherwise): such an update is statistically equivalent to first computing an average and then realizing the next opinion.  Based on this alternative description, the voters can be viewed as using an entirely discrete protocol, with binary states and discrete stochastic update procedures (neighbor selection followed by probabilistic realization of the next state).

The network interactions among the agents, whatever their types, are entirely specified by
the weights $g_{ij}$.  Thus, the row-stochastic matrix $G=[g_{ij}]$ is referred to as the {\em network matrix}.   The $i$th row of the network matrix, which indicates the weighting of neighbors' opinions by each agent, is denoted by ${\bf g}_i^T$ in our development.  Finally, the {\em opinion vector} is defined as ${\bf x}[k] =\begin{bmatrix} x_1[k] \\ \vdots 
\\ x_n[k] \end{bmatrix}$.  

The {\em network graph} $\Gamma$ is defined as a weighted digraph with $n$ vertices, which correspond to the $n$ agents in the network.  A directed edge is drawn from vertex $j$ to vertex $i$ (where self-loops are allowed) if and only if $g_{ij}>0$, and the weight of the edge is set to $g_{ij}$.    The graph $\Gamma$ illustrates which neighboring agents directly influence the next status of each agent.

\begin{figure} [!htb] 
	\centering
	\includegraphics[width=3.5in]{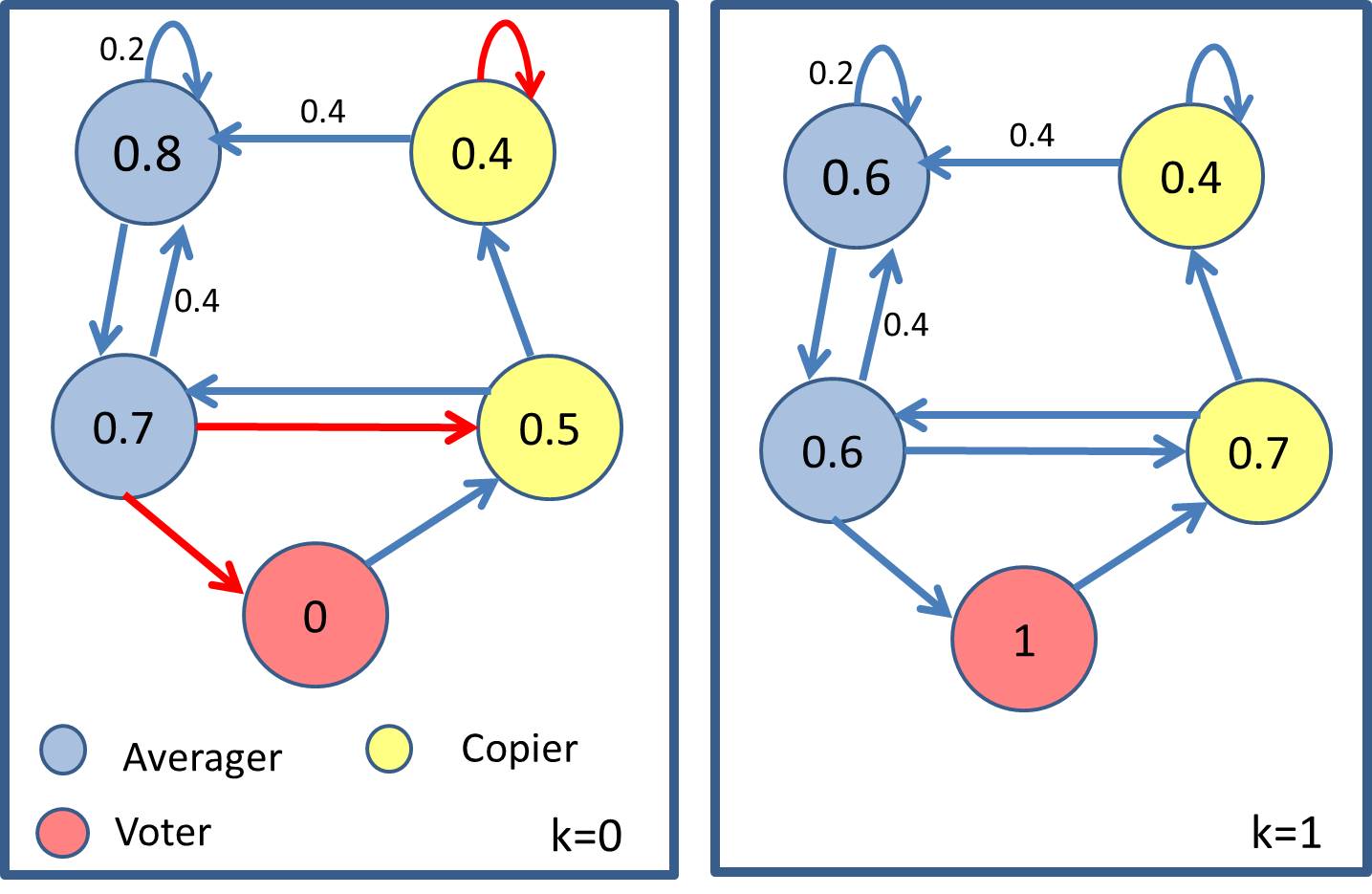}
	\caption{The state update for the HMAN is illustrated.  
	Averager agents compute a weighted average of neighbors' states, while copiers and voters
	update their states based on stochastically selected neighbors.}
    \label{fig:0}
\end{figure}

We refer to the dynamical network model as a whole as the {\bf hybrid multi-agent network} (HMAN),
see Figure \ref{fig:0} for an illustration.  
The HMAN encompasses several common models for synchronization or consensus.  The HMAN with only averagers is a classical deterministic linear model for network consensus  \cite{sync3,lincon1,lincon2,lincon3,tree}, and is also descriptive of social and physical-world diffusion processes (e.g., rumor or disease spread, heat exchange among chambers in an oven, or diffusion of a pollutant) \cite{colizza,scaled,incidental}.  Meanwhile, the case where all of the agents are either copiers or voters, and the agents' initial opinions are each binary ($0$ or $1$), 
resolves to the well-known voter model \cite{voter1,voter2,voter3}. The voter model is also a special instance of a broader class of quasi-linear discrete-valued network automata known as influence models \cite{influence1,influence2,chalee,roythesis}.

\section{Statistical Analysis and Conditions for Consensus}

A first question of interest is whether the HMAN achieves consensus when the network graph is connected, similarly to voter models and linear consensus/synchronization processes \cite{voter1,voter2,tree}.  
We primarily focus here on characterizing consensus from a mean-square-convergence perspective, as formalized in the following definition:

\begin{definition}
The HMAN is said to achieve mean-square consensus if 
$lim_{k \rightarrow \infty} E((x_i[k]-x_1[k])^2)=0$ for $i=2,\hdots,n$.
\end{definition}
Mean-square consensus implies that $lim_{k \rightarrow \infty} E((x_i[k]-x_j[k])^2)=0$ for any two agents $i$ and $j$.

 A two-moment statistical analysis of the HMAN's opinion vector is pursued, to enable development of graph-theoretic conditions for mean-square consensus.  The opinion vector turns out
to satisfy a moment-closure property,  like voter/influence models and infinite-server queueing network models  \cite{voter1,voter2,voter3,influence1,chalee,queue1,queue2}.  This closure property permits statistical analysis of the first and second moments via linear recursions.  Formally, let us define the {\bf extended opinion vector (EOV)} as ${\bf x}_{(e)}[k]=\begin{bmatrix} {\bf x}[k] \otimes {\bf x}[k] \\ {\bf x}[k] \end{bmatrix}$, where $\otimes$ refers to the Kronecker product.  The EOV
contains the agents' individual opinions and their pairwise products.  From the HMAN's update law,
the conditional expectation $E({\bf x}_{(e)}[k+1] \, | \, {\bf x}[k])$ can be readily computed as a linear function of ${\bf x}_{(e)}[k]$.  By further taking an expectation with respect to
${\bf x}[k]$, a recursion on the expected EOV is obtained:
\begin{equation}
E({\bf x}_{(e)}[k+1])=G_{(e)} E({\bf x}_{(e)}[k]),
\end{equation}
where recursion matrix has the form 
$G_{(e)}=\begin{bmatrix} G_2 & G_{21} \\ 0 & G \end{bmatrix}$.   To explicitly specify
the matrices $G_{2}$ and $G_{21}$, it is convenient to index each of their $n^2$ rows by the pair of agents $(i,j)$ whose next-state product expectation is specified by the row (thus, $(i,j)$ indexes the row $(i-1)n+j$).  The rows of $G_2$ and $G_{21}$ are given by the following:
\begin{itemize}
\item For $i \neq j$, row $(i,j)$ of $G_2$ equals $g_i^T \otimes g_j^T$, while row $(i,j)$ 
of $G_{21}$ is a zero vector.
\item  For $i=j$ corresponding to an averager agent, row $(i,j)$ of $G_2$ equals $g_i^T \otimes g_j^T$, while row $(i,j)$ of $G_{21}$ is a zero vector.
\item  For $i=j$ corresponding to a copier agent, row $(i,j)$ of $G_2$ has entries as follows:
the entry at column $(z-1)n+z$, where $z=1,\hdots, n$, is given by $g_{iz}$ while the remaining columns are zero.  Also, row $(i,j)$ of $G_{21}$ is zero.
\item For $i=j$ corresponding to a voter agent, row $(i,j)$ of $G_2$ is a zero vector, while row $(i,j)$ of $G_{21}$ equals ${\bf g}_i^T$.
\end{itemize}
The above specification of the matrix $G_{(e)}$ follows readily from first-principles analysis of the the conditional expectation of the EOV; details are omitted.

The statistical analysis of the EOV allows the development of graph-theoretic conditions for mean-square consensus.  The HMAN is found to achieve mean square consensus under broad connectivity conditions, which are comparable to those for the voter and linear consensus models, as formalized in the following theorem.

\begin{theorem}
The HMAN achieves mean-square consensus if the network graph (equivalently, network matrix) is ergodic\footnote{The network graph is called ergodic if it is strongly connected and also aperiodic, see e.g. \cite{gallager}.}.
\label{th:main}
\end{theorem}

\begin{proof}
An analysis of the recursion matrix $G_{(e)}$ in terms of the graph $\Gamma$ is needed, to 
verify mean-square consensus.  The graph-theoretic analysis is subtly different for a HMAN without voter agents, in which case $G_{21}=0$, as compared to a HMAN which has voter agents.  The proof uses standard results and terminology for stochastic and nonnegative matrices, see \cite{gallager,berman}.

The proof requires defining a new graph $\Gamma_2$ based on the $n^2\times n^2$ matrix $G_2$. Precisely, we define  $\Gamma_2$ as a weighted digraph with $n^2$ vertices, labeled as $q=1,\hdots,n^2$.  An edge is drawn from vertex $r$ to vertex $q$ in $\Gamma_2$ (where $q$ and $r$ may be identical) if and only if $[G_{2}]_{qr}$ is non-zero. Noticing that each vertex in $\Gamma_2$ corresponds to 
a pairwise product $x_i[k] x_j[k]$ in the EOV, it is convenient alternatively to label each vertex $q$ by the pair $(i,j)$ where $q=n(i-1)+j$.  For the pairwise labeling of vertices, an edge is indicated by two pairs in order, i.e.
there is an edge from $(\tilde{i},\tilde{j})$  to $(i,j)$ if $[G_{2}]_{q\tilde{q}}$ is non-zero for 
$q=n(i-1)+j$ and $\tilde{q}=n(\tilde{i}-1)+\tilde{j}$.

Let us now prove the result for the averager-copier model (the HMAN without voter agents).  For this case,
the matrix $G_{2}$ is  row-stochastic. Further, if the network matrix $G$ is ergodic, we claim that $G_{2}$ has a single recurrent class which is ergodic.  To see why, it is useful to first determine which edges are present in the second-moment-recursion graph
$\Gamma_2$ in terms of the network graph $\Gamma$.  First, from the algebraic expression of $G_{2}$, there is an edge
from the vertex $(\tilde{i},\tilde{i})$ to the vertex $(i,i)$ in $\Gamma_2$ if and only if there is an edge from vertex $\tilde{i}$ to vertex $i$ in $\Gamma$.  Next, let us consider whether there is an edge 
from an arbitrary vertex $(\tilde{i},\tilde{j})$ to a vertex $(i,j)$ in $\Gamma_2$.  If either 
$i \neq j$ or $i$ corresponds to an averager agent, then there is an edge in $\Gamma_2$ if and only if 
there is an edge from $\tilde{i}$ to $i$ and an edge from $\tilde{j}$ to $j$ in $\Gamma$.  In the special case where $i=j$ and $i$ corresponds to a copier agent, notice that there are only edges from other identical vertex pairs $(\tilde{i},\tilde{i})$ in $\Gamma_2$, and only if there is an edge from $\tilde{i}$ to $i$ in $\Gamma$. 

From the above description of edges in $\Gamma_2$, the following properties can be deduced.  First, the induced subgraph of $\Gamma_2$ on the subset consisting of the $n$ identical-pair vertices (vertices $(1,1), (2,2), \hdots, (n,n)$) has the identical connection structure to the graph $\Gamma$.  
Hence, it follows that this subgraph is strongly connected and ergodic.  
Next, let us show that there is a path from a vertex in this subgraph, say $(1,1)$, to any vertex $(i,j)$ in $\Gamma_2$.  Since the graph $\Gamma$ is ergodic, there exists a $\widehat{k}$ such that there is a path from vertex $1$ to vertex $i$ in $\Gamma$ of length $\widehat{k}$, and also a path of from vertex $1$ to vertex $j$ in $\Gamma$ of length $\widehat{k}$.  From the characterization of edges in $\Gamma$, it thus follows immediately that either there is a path from  $(1,1)$ to $(i,j)$ in $\Gamma_2$ of length $\widehat{k}$, or there is a path of length less than $\widehat{k}$ from  $(\tilde{i},\tilde{i})$ to 
$(i,j)$ where $\tilde{i}$ corresponds to a copier agent.  Thus, we have verified that there is a path from $(1,1)$ to every vertex $(i,j)$ in $\Gamma_2$ (either directly or via $(\tilde{i},\tilde{i})$).  In consequence, the graph $\Gamma_2$ necessarily has only one recurrent class.  Further, since this recurrent class has within it a subgraph that is ergodic, then the recurrent class is known to be ergodic.  (Note that the edge directions in our graphs $\Gamma$ and $\Gamma_2$ are reversed from those of a standard Markov chain; hence the recurrence terminology here is used in reverse from the standard.)

For the HMAN with at least one voter, the full recursion matrix 
$G_{(e)}$ is row-stochastic, and further can be shown to have a single recurrent class which is ergodic. This can be proved in a very similar way to that for the averager-copier model, albeit with some additional sophistication.  The proof is excluded to save space.

The graph-theoretic characterization of $G_{(e)}$ allows proof of mean-square consensus.  For the averager-copier model, the analysis of $G_{2}$ implies that it has a non-repeated eigenvalue $\lambda=1$ that is strictly dominant.  Further, the right eigenvector of $G_{2}$ associated with the unity eigenvalue is
${\bf v}={\bf 1}$, where ${\bf 1}$ indicates a column vector with all entries equal to $1$.  It follows immediately that $\lim_{k \rightarrow \infty} E({\bf x}[k] \otimes {\bf x}[k])=d {\bf 1}$, where 
the scalar $d$ depends on the initial opinion vector ${\bf x}[0]$.  For the HMAN with at least one voter, from the fact that ${G}_{(e)}$ is a row-stochastic matrix with a single recurrent class that is ergodic, it follows (using the same argument as for the averager-copier model) that 
\begin{small} $\lim_{k \rightarrow \infty} E( {\bf x}_{(e)}[k])=d {\bf 1}$, \end{small} where the scalar $d$ again depends on the initial opinions.  
Thus, for any HMAN model, we have that  $\lim_{k\rightarrow \infty} E((x_j[k]-x_1[k])^2)$ is $0$, for any $j=1,\hdots,n$.   $\square$
\end{proof}

Two further remarks about the consensus result are worthwhile.

\noindent {\em Remark 1:} The expected squared difference $ E((x_j[k]-x_1[k])^2)$ approaches zero exponentially with respect to  $k$,   and hence the probability that $|x_j[k]-x_1[k]|<\epsilon$  also approaches $1$ exponentially with  $k$, for any positive $\epsilon$.  Based on this exponential convergence in probability together with state boundedness,  the sequences $x_j[k]-x_1[k]$ can also be shown to converge to zero almost everywhere.  Thus, consensus is also achieved in an almost-everywhere sense.

\noindent {\em Remark 2:} While the averager-copier model and the general HMAN model both achieve consensus, their asymptotic behaviors are qualitatively different.  The agents in the averager-copier model may asymptote to any real number on the interval $[0,1]$, depending on the
network's topology and the initial opinions.  If even one voter agent is present, however, the 
agents may only asymptote to $0$ or $1$, since the voter agents are constrained to have binary opinions.

The consensus dynamics of the HMAN are illustrated using an example with five agents.  In the example, Agents $1$-$3$ are averagers, Agent $4$ is a copier, Agent $5$ is a voter, and the network matrix is \begin{small} $G=\begin{bmatrix} .4 & .2 &.2 & 0 & .2 \\
.2 & .4 & .2 & .2 & 0  \\.2 & .2 & .4 & .2 & 0 \\ 0 & .2 &.2 & .6 & 0 \\ .2 & 0 & 0 & 0 &.8 \end{bmatrix}$. \end{small}  
A simulation (sample path) of this HMAN is shown in Figure \ref{fig:1}, for a random initial state.  The agents' opinions asymptotically converge to $1$ for this particular simulation. The copier and voter agents' opinion trajectories look qualitatively different from the averagers' trajectories: they are more jumpy, and the voter agent's opinion is restricted to be binary. The mean square deviations between Agent $2$ and Agent $1$, Agent $4$ and Agent $1$, and Agent $5$ and Agent $1$ (i.e. $E((x_j[k]-x_1[k])^2)$ for $j=2,4,5$) are computed, and plotted vs time in Figure 
\ref{fig:2}.  These expected errors (deviations) converge to $0$ asymptotically, as guaranteed by Theorem $1$.  The error between the two averager agents diminishes quickly, while the error between the averager and voter agents decays slowly, and the error between the averager and copier agent has an intermediate decay rate.

\begin{figure} [!htb] 
	\centering
	\includegraphics[width=3.7in]{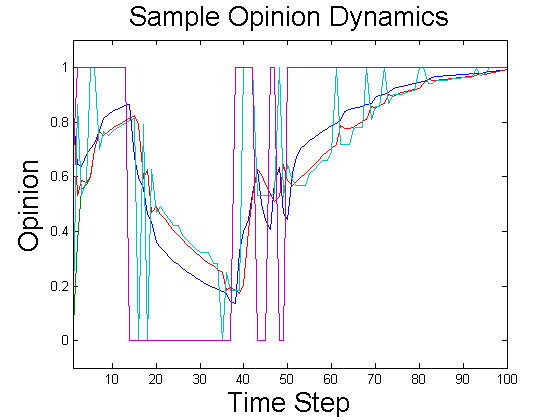}
	\caption{ The agents' opinion trajectories  for a HMAN with five agents (three averagers, one copier, one voter).  Note that the opinion of the voter agent is restricted to $0$ and $1$, while the copier's opinion is distinguished by its jumpiness. Eventually, the HMAN reaches consensus at $1$.}
    \label{fig:1}
\end{figure}

\begin{figure} [!htb] 
	\centering
	\includegraphics[width=3.7in]{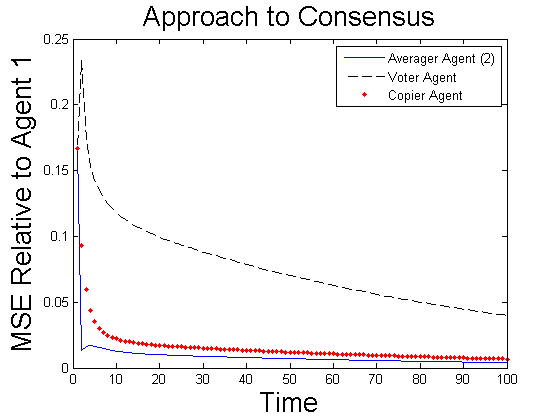}
	\caption{The mean square deviation (error) of other agents from Agent $1$ is plotted vs. time.  
	The errors converge to $0$, however the rates of convergence vary among the different agent types.}
	\label{fig:2}
\end{figure}

\section{Bounds on the Consensus Time}

A main focus of the research on consensus processes has been to characterize or bound the consensus time.  For instance, a number of graph-theoretic results have been developed for the consensus time of the voter and linear consensus models \cite{voter2,lincon1,cox}.  Here, we study the consensus time for HMAN, for the case that  the network graph is ergodic.
For the HMAN, the mean and distribution of the consensus time are closely tied to the spectrum of the moment-recursion matrix $G_{(e)}$.  
Since $G_{(e)}$ is a row-stochastic matrix, it necessarily has dominant eigenvalue $\lambda=1$. Further, the {\em subdominant eigenvalue} $\lambda_s$
of $G_{(e)}$ (the largest-magnitude eigenvalue whose magnitude is less than $1$ or strictly inside the unit circle) primarily governs the worst-case mean consensus time and consensus-time distribution.  Precisely, the distance of the subdominant eigenvalue to the unit circle relates inversely to the worst-case mean settling time (i.e. the further from the unit circle, the faster it reaches consensus).  In fact, an exponential bound on the deviation from consensus and the tail of the consensus-time distribution can be found in terms of $\lambda_s$ (see e.g. \cite{aldous,aldous2}).

The efficiency of consensus via different update mechanisms (averaging, voting, copying) can be compared, by comparing the location of the subdominant eigenvalue for different HMAN models with the same network matrix.  A first simple outcome  is that the averager model, where all agents use averaging, reaches consensus at least as fast as any other HMAN with the same network graph. Further, if the network has at least one voter agent, then the consensus rate is strictly slower than for the  averager model, except for special star-like network graphs.  

The main insight needed to formalize the above result is that the subdominant eigenvalue $\lambda_2$ of 
$G$ is also an eigenvalue of $G_{2}$ which is however not dominant, for any HMAN.  For the averager model, we have that $G_{2}= G \otimes G$.  From the standard spectral analysis of the Kronecker product, it follows that $\lambda_2$ is also the subdominant eigenvalue of $G_2$.  For any other HMAN, the matrix $G_{2}$ is a modification of the matrix $G \otimes G$, where only the rows indexed by identical vertex pairs $(i,i)$ corresponding to copier or voter agents (i.e., rows
$(i-1)N+i$ where $i$ is a copier or voter) are modified.  Further, for these rows corresponding to copiers, only the columns of $G_{2}$ indexed by identical vertex pairs are non-zero; for those corresponding to voters, the entire row is zero in $G_2$.  To continue, consider the eigenvector ${\bf v}_2$ of $G$ associated with $\lambda_2$.  We notice that ${\bf v}={\bf 1} \otimes {\bf v}_2 - {\bf v}_2 \otimes {\bf 1}$  is an eigenvector of $G \otimes G$ for the eigenvalue $\lambda_2$.  Using the additional fact that the entries in ${\bf v}$ indexed by identical vertex pairs are $0$, we thus find that ${\bf v}$ is also a right eigenvector of $G_{2}$ with eigenvalue $\lambda_2$ (since the product $G_{2} {\bf v}$ is identical to $(G \otimes G) {\bf v}$.  Noticing further that: 1) $G_{2}$ is a nonnegative matrix and 2) ${\bf v}$ has both positive and negative entries, it follows immediately from the Frobenius-Perron theorem that $\lambda_2$ is not the dominant eigenvalue of $G_{2}$.  The subdominant eigenvalue of $G_{2}$  is thus at least as large in magnitude as $\lambda_2$, and hence so is the subdominant eigenvalue of $G_{(e)}$.  

In the case that the model
has at least one voter agent, the {\em dominant} eigenvalue of $G_{2}$ is in fact the subdominant eigenvalue of $G_{(e)}$.  Hence, it follows that the subdominant eigenvalue of $G_{(e)}$ must be strictly larger than $\lambda_2$, except in the very special case that $G_{2}$ contains a periodic class with multiple eigenvalues of magnitude $|\lambda_2|$.  It is easy to check that this only happens in star-like networks, wherein the voter agents' opinions are distributed to the remaining agents with certainty in a finite number of time steps.   

Thus, we have shown that the general HMAN reaches consensus no faster than the pure averager model for any specified network graph, and in fact is strictly slower if any voter agents are present (except for some very special topologies).  This result generalizes the observation that voter models reach consensus slower than averaging models, which has been previously characterized via the analysis of coalescing random walks \cite{aldous,aldous2,coalesce}.

Interestingly, while the voter model reaches consensus slower than the averager model for a given network matrix, mixed averager-copier-voter models are even slower provided that there is at least one voter agent.  This follows immediately from the fact that the matrix $G_{2}$ for the mixed model is the sum of the matrix $G_2$ for the voter model and another nonnegative matrix.  Thus, the dominant eigenvalue of $G_2$ for the mixed model, which is the subdominant eigenvalue of $G_{(e)}$, is at least as large in magnitude as for the pure voter model (and is strictly larger except in some very special cases).  Thus, the presence of agents with graded opinions slows the convergence of binary decision-making processes.

Bounds on the consensus time can also be developed based on the fact that $G_{2}$ is a modification of the Kronecker product $G \otimes G$.    
A full development is outside the scope of the brief paper, but one general bound is obtained and its tightness is evaluated from a graph-theoretic perspective, for voter-averager models on an undirected graph. Specifically, a lower bound is obtained
on the subdominant eigenvalue $\lambda_s$, which is the dominant eigenvalue of $G_2$ for the voter-averager model.  For this special case (voter-averager only), this eigenvalue is equivalently the dominant eigenvalue of the principal submatrix $\widehat{G}_2$ of $G \otimes G$, which is formed by deleting the rows and columns corresponding to  identical vertex pairs $(i,i)$ for the $m_v$ voter agents.  Since $\widehat{G}_2$ is a symmetric nonnegative matrix, its dominant eigenvalue is lower-bounded by its average row sum \cite{johnson}.  
Notice that  $\widehat{G}_2$ has $n^2-m_v$ rows.  Also, from the fact that $G \otimes G$ is a symmetric stochastic matrix, the sum of the entries in $\widehat{G}_2$ is necessarily lower bounded by $n^2-2m_v$.  It follows immediately that $\lambda_s \ge \frac{n^2-2m_v}{n^2-m_v}= 1-\frac{m_v}{n^2-m_v}$.  
Thus, the average time to consensus (which equals $\frac{1}{\vert 1-\lambda_s \vert}$) is at least $\frac{n^2-m_v}{m_v}=\frac{n^2}{m_v}-1$.  Thus, consensus in large voter-averager networks with few voter agents is seen to be quite slow.  

The above bound on the consensus rate $\lambda_s$  is tight if: 1) the average row sum of $\widehat{G}_2$ equals its dominant eigenvalue and 2) the voter agents are mutually independent (no edges between any pair). 
The first condition holds if the graph is regular,  while the second holds if the voters are sparsely placed.  Thus, for instance, large Erdos-Renyi random graphs \cite{erdos} with sparse randomly-placed voters achieve the bound.  In particular, consider a network graph $\Gamma$ which is an Erdos-Renyi random graph with connection probability $p$, with identical edge weights between connected vertices (and self-loops chosen to make the row sums $1$).  The consensus rate bound is tight in the limit of large $n$ provided that: $1 << np << \frac{n}{m_v}$.
In Figure 3, the ratio between the average consensus time $\frac{1}{\vert 1-\lambda_s\vert}$ and the bound on the consensus time $\frac{n^2}{m_v}-1$ is shown as a function of the network size $n$, for $p=0.2$ and $m_v=1$.

\begin{figure} [!htb] 
	\centering
	\includegraphics[width=3.7in]{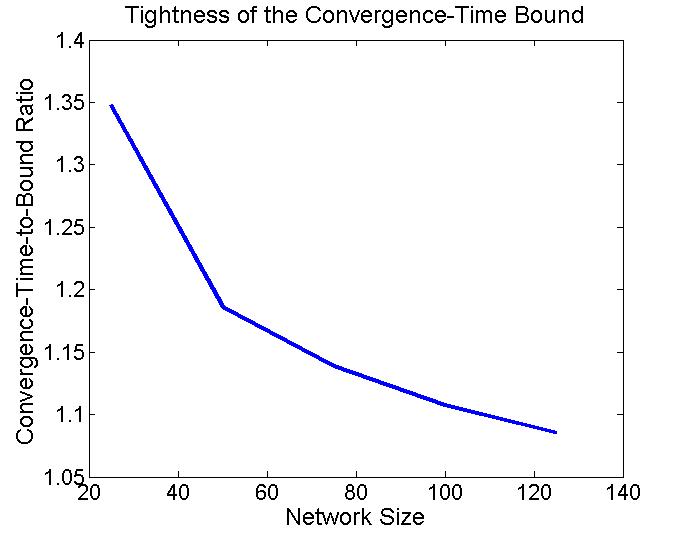}
	\caption{The ratio between the average time to consensus and the derived bound is shown, for Erdos-Renyi random graphs of different sizes.  The bound becomes tight for large graphs.}
	\label{fig:3}
\end{figure}

\end{document}